\documentclass[a4paper,UKenglish,cleveref, autoref, thm-restate]{lipics-v2021}

\author{Zsuzsanna Lipt{\'a}k}{Department of Computer Science, University of Verona, Italy}{zsuzsanna.liptak@univr.it}{https://orcid.org/0000-0002-3233-0691}{}

\author{Francesco Masillo}{Department of Computer Science, University of Verona, Italy}{francesco.masillo@univr.it}{https://orcid.org/0000-0002-2078-6835}{}

\author{Gonzalo Navarro}{CeBiB \& Department of Computer Science, University of Chile, Chile}{gnavarro@dcc.uchile.cl}{https://orcid.org/0000-0002-2286-741X}{}

\title{Maintaining the cycle structure of dynamic permutations} 

\authorrunning{Zs.\ Lipt\'ak, F.\ Masillo, and G.\ Navarro} 

\Copyright{Zsuzsanna Lipt\'ak, Francesco Masillo, and Gonzalo Navarro} 



\date{February 2023}
\category{} 

\relatedversion{} 

\supplement{}

\nolinenumbers 

\hideLIPIcs  

\EventEditors{John Q. Open and Joan R. Access}
\EventNoEds{2}
\EventLongTitle{42nd Conference on Very Important Topics (CVIT 2016)}
\EventShortTitle{CVIT 2016}
\EventAcronym{CVIT}
\EventYear{2016}
\EventDate{December 24--27, 2016}
\EventLocation{Little Whinging, United Kingdom}
\EventLogo{}
\SeriesVolume{42}
\ArticleNo{23}

\usepackage{multicol}
\usepackage{amsthm}
\usepackage{latexsym,amsmath,amssymb,stmaryrd,mathtools}
\usepackage{graphicx,epsfig,tikz,color,xcolor,bbding,algorithm2e}
\usepackage{float}
\usetikzlibrary{shapes.geometric,positioning,chains,fit,calc}
\tikzset{near start abs/.style={xshift=1cm},
	node/.style={circle,draw},
	nodeone/.style={circle,draw,gray, line width=0.7mm},
	nodered/.style={circle,draw,red, line width=0.7mm}}

\newcommand{\Oh}{{\cal O}}

\newcommand{\FST}{{\text{FST}}}

\newcommand{\dist}{\mathop{\text{dist}}}

\begin{document}

\maketitle

\begin{abstract}
We present a new data structure for maintaining dynamic permutations, which we call a {\em forest of splay trees (FST)}. The FST allows one to efficiently maintain the cycle structure of a permutation $\pi$ when the allowed updates are transpositions. The structure stores one conceptual splay tree for each cycle of $\pi$, using the position within the cycle as the key. Updating $\pi$ to $\tau\cdot\pi$, for a transposition $\tau$, takes $\Oh(\log n)$ amortized time, where $n$ is the size of $\pi$. The FST computes any $\pi(i)$, $\pi^{-1}(i)$, $\pi^k(i)$ and $\pi^{-k}(i)$, in $\Oh(\log n)$ amortized time. Further, it supports cycle-specific queries such as determining whether two elements belong to the same cycle, flip a segment of a cycle, and others, again within $\Oh(\log n)$ amortized time.
\end{abstract}


\newpage

\setcounter{page}{1}

\section{Introduction}\label{sec:introduction} 

Permutations play a central role in many applications in mathematics and computer science, spanning from combinatorics and group theory to algorithms for random generation and computational biology (most notably, genome rearrangements). Permutations are also of fundamental importance in text indexing: a {\em text index} is a dedicated data structure that allows fast execution of pattern matching queries and other text processing tasks. 
All current compressed and non-compressed text indexes, such as suffix arrays (SAs)~\cite{ManberM93}, compressed suffix arrays (CSAs)~\cite{GrossiV05}, compressed suffix trees (CSTs) \cite{GNPjacm19,FMNtcs09}, the FM-index~\cite{FM05}, the RLFM-index~\cite{MakinenN05}, the $r$-index~\cite{GNPjacm19}, and the LZ-index~\cite{KreftN13}, to name a few, have at their core some permutation (see, e.g., \cite[Ch.\ 11 \& Section\ 13.3]{NavarroBook}).
The best known of these, the SA, is itself a permutation of $n$; other permutations used by text indexes include those known as $\Psi$, LF (a.k.a.\ standard permutation), or $\phi$. 
Another prominent example where permutations are central is genome rearrangements in computational biology~\cite{SM97,BookGenomeRearrangments}, which consists of {\em sorting} a permutation: Given a set of legal operations and a permutation $\pi$, find a shortest sequence of operations that transforms $\pi$ into the identity.

Permutations, in most situations, are simply stored as an array of integers, allowing constant-time access to each entry $\pi(i)$ of the permutation $\pi$. Sometimes one needs to take the inverse or powers of the permutation, which can be done efficiently with little additional space \cite{MRRR12,BN13,BCGNN14}. Other queries of interest, such as determining if two elements belong to the same cycle, may also be supported efficiently. There is not much support, instead, to let $\pi$ change between queries. In this case one needs a {\em dynamic data structure} that can be efficiently updated when some operation is applied to the permutation to convert it into another permutation $\pi'$. 

A simple workaround is to regard the permutation as a sequence and use a representation for dynamic sequences \cite{MunroN15}. Such a representation can support some simple queries on $\pi$, but in some applications (e.g., \cite{GiulianiLMR21}), the {\em cycle structure} of the permutation needs to be kept track of along updates, including the number and cardinality of the cycles, or deciding whether two elements are in the same cycle.

In this paper we fill this gap, presenting a data structure tailored for maintaining dynamic  permutations called FST (forest of splay trees), with special focus on their cycle structure. The FST supports arbitrary transpositions and flips in the cycles, while answering the described queries and some others, all in $\Oh(\log n)$ amortized time. Indeed, the FST is the only structure to date able to solve both those queries and updates within $o(n)$ time.


\subsection{Related work}\label{sec:related} 

Given a permutation $\pi$ on $n$, $\Oh(1/\epsilon)$-time 
support for accessing any $\pi(i)$, its inverse $\pi^{-1}(j)$, and in general positive and negative powers of $\pi$ (i.e., $\pi^k(i)$ or $\pi^{-k}(j)$ for any integer $k>0$), can be obtained with a representation using $(1+\epsilon) n\log n$ bits of space, for any $\epsilon>0$ \cite{MRRR12}.\footnote{By default, our logarithms are in base 2.} This is, in principle, an array storing the consecutive values not of $\pi$, but of its sequence of cycles seen as another permutation $\rho$. For example, if the cycle decomposition of $\pi$ is $(1,3)(2,6,5)(4)$, then $\rho = 1,3,2,6,5,4$. They provide $\Oh(1)$-time access to $\rho(i)$ and $\Oh(1/\epsilon)$-time access to $\rho^{-1}(j)$. They also store a bitvector of length $n$ marking with 1s the places where cycles end in $\rho$, with support to find, in constant time, the 1s that precede and follow any given position. In our example, this bitvector is $0,1,0,0,1,1$. This structure allows them to compute any positive or negative power of $\pi$ (including $\pi$ and $\pi^{-1}$) in time $\Oh(1/\epsilon)$. Other representations aim to store some classes of permutations in compressed form while supporting the same functionality \cite{BN13,BCGNN14}.

None of those constructions support updates to $\pi$, however. 
A simple way to support updates and some queries is to regard the one-line representation of $\pi$ as a string $P$ of length $n$ over the alphabet $\{1,\ldots,n\}$, where each symbol appears exactly once. A sequence representation of $P$ that supports accessing any $P[i]$ and finding the first occurrence of any value $j$ in $P$ (called $select_j(P,1)$ in the literature) simulates the queries $\pi(i) = P[i]$ and $\pi^{-1}(j) = select_j(P,1)$. Dynamic sequence representations \cite{MunroN15} applied on $P$ then support both operations in time $\Oh(\frac{\log n}{\log\log n})$ and use $n\log n + o(n\lg n)$ bits. Within the same time, they support insertions and deletions of symbols in $P$, which can be used to reflect the updates we wish to perform on $\pi$. An insertion is an operation that, given as input a sequence $S$ of length $n$, a position $i$, and a character $c$, gives as output a new sequence $S'$ where $S'[1, i-1] = S[1, i-1]$, $S'[i] = c$, and $S'[i+1, n+1] = S[i, n]$. A deletion is an operation that, given as input a sequence $S$ of length $n$ and a position $i$, returns as output a new sequence $S'$ where $S'[1, i-1] = S[1,i-1]$ and $S'[i, n-1] = S[i+1, n]$. Other more sophisticated queries are supported only naively, however, for example arbitrary powers $\pi^k$ and $\pi^{-k}$ take time $\Oh(k \cdot \frac{\log n}{\log\log n})$. In general, such a representation does not efficiently support important queries and update operations regarding the cycle structure of $\pi$.


\subsection{Our contribution}\label{sec:contributions}

We introduce a new data structure, called {\em forest of splay trees (FST)}, which is to the best of our knowledge the most efficient one to maintain dynamic permutations when the focus of interest is the permutation's cycle structure. The FST supports the following operations, which significantly enhance the data structure presented by Giuliani et al.~\cite{GiulianiLMR21}.

\begin{itemize}
    \item Access: compute $\pi(i)$ and $\pi^{-1}(j)$ for any $i$ and $j$.
    \item Powers: compute $\pi^k(i)$ and $\pi^{-k}(j)$ for any $i$, $j$, and integer $k$.
    \item Number of cycles: give the number of cycles in $\pi$.
    \item Cycle size: give the number of elements in the cycle of $i$, for any $i$.
    \item Same cycle: determine if $i$ and $j$ are in the same cycle, for any $i$ and $j$.
    \item Distance: return $d$ such that $\pi^d(i)=j$, if such a $d$ exists, $\infty$ otherwise, for any $i$ and $j$. 
    \item Transpose: exchange the values $i$ and $j$ in the one-line representation of $\pi$ (``transpose $i$ and $j$''), or exchange the values at positions $i$ and $j$ (``transpose $\pi(i)$ and $\pi(j)$'').
    \item Flip: reverse a segment of the cycle to which both values $i$ and $j$ belong.
\end{itemize} 

The following theorem summarizes our contribution.

\begin{theorem}
Let $\pi$ be a permutation on $n$. The FST is built from $\pi$ in $\Oh(n)$ time and uses $3n\log n+\Oh(n)$ bits of space. Once built, it supports each of the queries and updates above in $\Oh(\log n)$ amortized time. 
\end{theorem}

The original data structure, which we now generalize with the FST, supports only a restricted class of transpositions on permutations with specific properties. It already proved useful on an application using the Burrows-Wheeler Transform \cite{Burrows94}, where it allowed replacing an $\Oh(n^2)$ algorithm with an $\Oh(n \log n)$ time one \cite{GiulianiLMR21}. 

The rest of the paper is organized as follows: In Section~\ref{sec:basics}, we give the necessary definitions. In Section~\ref{sec:FST}, we present the \FST\ data structure and  the operations. In Section~\ref{sec:comparison}, we give a comparison of the \FST\ to existing data structures for permutations. We close with an outlook in Section~\ref{sec:conclusion}.


\section{Basics} \label{sec:basics}

\subsection{Permutations}

Given a positive integer $n$, a {\em permutation of $n$}  is a bijection from the set $\{1,2,\ldots,n\}$ to itself. Two common ways to represent a permutation $\pi$ are the two-line notation $\left(\begin{smallmatrix} 
1 & 2 & \ldots & n\\
\pi(1) & \pi(2) & \ldots & \pi(n)
\end{smallmatrix}\right)$ 
and the one-line notation, $\pi(1)\pi(2)\cdots\pi(n)$, where the top row is omitted. The set of permutations of $n$ forms a group together with the composition (or product) of functions, and is called the {\em symmetric group}, denoted $S_n$. We write $\tau \cdot \pi$ for the permutation resulting from applying $\tau$ after $\pi$, that is, $(\tau \cdot \pi)(i) = \tau (\pi(i))$. 
A {\em cycle} of $\pi$ is a minimal subset $C$ of $\{1, 2, \ldots, n\}$ such that $\pi(C) = C$. A permutation consisting of only one cycle is said to be {\em cyclic}, a cycle of length 1 is a {\em fixpoint}, and a cycle of length 2 is called a {\em transposition}. Every permutation can be uniquely decomposed into disjoint cycles ({\em cycle decomposition}), and can thus be represented as a composition of its cycles. For example, the cycle representation of the permutation $\pi = \left(\begin{smallmatrix} 1 & 2 & 3 & 4 & 5 & 6 \\ 3 & 6 & 1 & 4 & 2 & 5  \end{smallmatrix}\right)$ is $(1, 3)(2,6,5)(4)$: of the three cycles, $(1, 3)$ is a transposition, $(2, 6, 5)$ is a cycle of length 3, and $(4)$ is a fixpoint. A cycle can be written starting from any one of its elements, for example $(2,6,5) = (6,5,2) = (5,2,6)$. Further, permutations in general do not commute, but disjoint cycles do. Thus, the cycle representation is unique only up to order of the cycles and up to cycle rotations, for example the above permutation could also be written as $\pi = (3,1)(4)(6,5,2)$. It is common to drop  fixpoints from the cycle representation, thus an $n$-permutation $\tau$ with $n-2$ fixpoints and one transposition $(i,j)$ is written simply as $\tau=(i,j)$. 

For more on permutations, see the book by B\'ona~\cite{BonaBook}, or any textbook on discrete mathematics, such as that of Aigner~\cite{AignerBook}. 

Let $\pi$ be a permutation of $n$ and $\tau$ a transposition. It is a well-known fact that the number of cycles of $\pi' = \tau \cdot \pi$ increases by $1$ if the two elements are in the same cycle, and decreases by $1$ if they are in different cycles. In particular, the exact form of the resulting permutation is given by the following lemma. For convenience, the transposition is given in the form $(\pi(x),\pi(y))$, for some $x$ and $y$. Note that an element is always in the same cycle as its image, so $x$ and $y$ are in the same cycle if and only if $\pi(x)$ and $\pi(y)$ are in the same cycle. 

\begin{lemma}[\cite{GiulianiLMR21}, Lemma 3]\label{lemma:cycles} 
	Let $\pi = C_1 \cdots C_k$ be the cycle decomposition of the permutation $\pi$, $x\neq y$, and $\pi' = (\pi(x),\pi(y)) \cdot \pi$. 
	\begin{enumerate}
		\item (Split case) If $x$ and $y$ are in the same cycle $C_i$, then this cycle is split into two. In particular, let $C_i = (c_1, c_2, \ldots, c_j, \ldots, c_m)$, with $c_m = x$ and $c_j = y$. Then $\pi' =$ $ (c_1,c_2,\ldots, c_{j-1},y)$ $(c_{j+1}, \ldots, c_{m-1},x) \prod_{\ell\neq i} C_{\ell}$.  
		
		\item (Join case) If $x$ and $y$ are in different cycles $C_i$ and $C_j$, then these two cycles are merged. In particular, let $C_i = (c_1, c_2, \ldots, c_m)$, with $c_m = x$, and $C_j = (c'_1, c_2', \ldots, c_r')$, with $c_r'=y$, then $\pi' = (c_1, \ldots, c_{m-1},x, c_1', \ldots, c'_{r-1},y) \prod_{\ell\neq i,j} C_{\ell}$. 
		
	\end{enumerate}
\end{lemma}

\begin{example} 
Let $\pi = \left(\begin{smallmatrix} 1 & 2 & 3 & 4 & 5 & 6 & 7 & 8 & 9 \\ 8 & 2 & 5 & 3 & 1 & 7 & 9 & 4 & 6  \end{smallmatrix}\right) = (1,8,4,3,5)(2)(6,7,9)$, then

$(\pi(1),\pi(4))\cdot\pi = \left(\begin{smallmatrix} 1 & 2 & 3 & 4 & 5 & 6 & 7 & 8 & 9 \\ \textcolor{red}{3} & 2 & 5 & \textcolor{red}{8} & 1 & 7 & 9 & 4 & 6  \end{smallmatrix}\right) = (1,3,5)(2)(4,8)(6,7,9)$ [Split case] 

$(\pi(3),\pi(6))\cdot\pi = \left(\begin{smallmatrix} 1 & 2 & 3 & 4 & 5 & 6 & 7 & 8 & 9 \\ 3 & 2 & \textcolor{red}{7} & 8 & 1 & \textcolor{red}{5} & 9 & 4 & 6  \end{smallmatrix}\right) = (1,3,7,9,6,5)(2)(4,8)$ [Join case]
\end{example}


\subsection{Splay Trees}\label{sec:splaytrees}

Splay trees~\cite{SleatorT85} are binary search trees that allow joining and splitting of trees in addition to the usual operations (such as access to, or insertion and deletion of, items). They are not necessarily balanced, but they undergo a self-adjusting procedure after each operation that guarantees that operations take amortized logarithmic time in the total number of items. In particular, 
after each operation involving an element $x$, this element $x$ is moved to the root of its tree via a structural rearrangement called {\em splaying}. Splaying consists of a series of edge rotations applied repeatedly on $x$ until it becomes the root of the tree. The rotations can be of one of three types, referred to as {\em zig}, {\em zig-zig}, and {\em zig-zag}, depending on the relative position of $x$ w.r.t.\ its parent and grandparent. For more details, see the Appendix and the original paper. 

Concretely, splay trees support the following operations: 

\begin{itemize} 
\item {\em access$(i,t)$}: return a pointer to item $i$ if it is in tree $t$, otherwise return NIL 
\item {\em insert$(i,t)$}: insert item $i$ into $t$ (assuming it is not present)
\item {\em delete$(i,t)$}: delete item $i$ from $t$ (assuming it is present)
\item {\em join$(t_1,t_2)$}: construct tree $t$ containing all items in $t_1$ and $t_2$ (assuming that all items in $t_1$ are strictly smaller than those in $t_2$) 
\item {\em split$(i,t)$}: return two new trees $t_1$ and $t_2$, where $t_1$ contains all items in $t$ less than or equal to $i$, and $t_2$ contains all items in $t$ greater than $i$
\end{itemize}

The next theorem from the original paper gives the complexity of the operations.

\begin{theorem} [Balance Theorem with Updates~{\cite[Thm.~6]{SleatorT85}}] \label{thm:SleatorT85}
	A sequence of $m$ arbitrary operations on a collection of initially empty splay trees takes $\Oh(m+ \sum^{m}_{j=1} \log n_j)$ time, where $n_j$ is the number of items in the tree or trees involved in operation $j$.
\end{theorem}


\section{Forest of splay trees (FST)} \label{sec:FST} 

Our data structure, FST, stores a forest of splay trees partitioning the set $\{1,\ldots,n\}$. The FST allows constant-time access to each node in addition to the usual operations on splay trees. The permutation $\pi$ is represented by one splay tree for each of its cycles. 

The elements are keyed by their position in their cycle, that is, the splay tree is a binary search tree with respect to the positions in the cycle (and not w.r.t.\ the elements themselves); see Figure~\ref{fig:FST-example} for an example. In particular, an in-order traversal of one of the splay trees yields the corresponding cycle. The fact that the linearization of a cycle is not unique (i.e., that all rotations represent the same cycle) will be important later. 

The concrete FST data structure consists of a  counter {\em cycles} and a 
$3\times n$ matrix $M$; for reasons of presentation we give a variant of the matrix with $4$ rows and show later how to reduce this to only $3$ rows. The counter {\em cycles} contains the number of cycles of $\pi$, while, for $1\leq i \leq n$, $M_{1i}$ is the parent of $i$ (NIL if $i$ is the root of its tree), $M_{2i}$ is the left child and $M_{3i}$ the right child of $i$ (NIL if no left, resp.\ right, child is present), and $M_{4i}$ is the size of the subtree rooted in $i$. See Figure~\ref{fig:FST-example} again for an example. 

\input{example_FST} 

We can reduce the matrix to only 3 rows by applying a standard trick on binary trees. We substitute the first three rows with only two rows containing the left child and right sibling information. If the right sibling does not exist, then the entry points to the parent. Thus, one can access the right child by accessing the left child and then its right sibling (2 steps); and the parent by accessing the right sibling: if there is none, then we get the parent immediately, otherwise we access the right sibling of the node pointed to. In both cases, we have replaced one step by one or two steps, thus, navigation time remains within the same bounds. 
The space required by the \FST\ is then $3n\lceil \log n\rceil = 3n\log n + \Oh(n)$ bits.


\subsection{Construction and splay tree operations}

Given an input permutation $\pi$, we compute the cycle decomposition of $\pi$ and assign to $c$ the number of cycles. We then build a splay tree for each cycle. Instead of building it by inserting all the cycle elements in an initially empty splay tree, which would lead to $\Oh(n\log n)$ construction time, we apply a standard linear-time construction that builds a perfectly balanced binary search tree for each cycle. We now show that the potential function of the splay trees created with that shape is also $\Oh(n)$, which allows us combine this $\Oh(n)$ construction time with all the other amortized operation times. 

\begin{lemma}\label{lemma:balanced}
The potential function of a perfectly balanced splay tree with $r$ nodes is $2r+\Oh(\log^2 r) \subseteq \Oh(r)$.

\end{lemma}

\begin{proof}
Let $d$ be the depth of the deepest leaves in a perfectly balanced binary tree, and call $\ell = d-d'+1$ the {\em level} of any node of depth $d'$. It is easy to see that there are at most $1+r/2^\ell$ subtrees of level $\ell$. Those subtrees have at most $2^\ell-1$ nodes. The potential function used in the analysis of splay trees \cite{SleatorT85} is
$\phi = \sum_v \log s(v)$, where $s(v)$ is the number of nodes in the subtree rooted at $v$ and the sum ranges over all the nodes of the splay tree. Separating this sum by levels $\ell$ and using the bound $s(v) < 2^\ell$ if $v$ is of level $\ell$, we get
$$\phi ~<~ \sum_{\ell=1}^{\log r} \left(1+\frac{r}{2^\ell}\right) \log 2^\ell ~=~ 
2r + \Oh(\log^2 r).$$
\end{proof}

Since all the splay trees together add up to $n$ nodes, the potential function $\phi$ of the forest is $\Oh(n)$ after building them all from the permutation. Further, since each splay tree contains at most $n$ nodes at any given time, we can derive the following corollary.

\begin{corollary} 
	A sequence of $m$ arbitrary splay tree operations (access, insertion, deletion, split, or join) on an \FST\ of an $n$-permutation takes $\Oh(n + m \log n)$ time.
\end{corollary}
\begin{proof} 
An edge rotation can be implemented on the \FST\ with 6 updates in the matrix $M$ (update cells $M_{1x}$, $M_{2x}$, $M_{1y}$, $M_{3y}$, $M_{1M_{2y}}$, $M_{2M_{1x}}$ or $M_{3M_{1x}}$ for a right rotation; it is similar for a left rotation). Regarding the {\em join} and {\em split} operations, a removal or addition of an edge implies changing 2 cells in $M$. The {\em cycles} counter has to be incremented for each {\em split} operation and decremented for each {\em join} operation, adding constant time in each case. Furthermore, after each edge rotation and {\em join} and {\em split} operations, the subtree sizes of the involved nodes need to be updated. This is done in the standard way for binary trees that are annotated with subtree sizes. We illustrate an edge rotation in Figure~\ref{fig:size-update}. Altogether, the splay tree operations {\em access, join}, and {\em split} can be implemented in constant time on the \FST. 

It follows from Lemma~\ref{lemma:balanced} that the original splay trees can be constructed in $\Oh(n)$ time and makes us start with the potential function at $\phi = \Oh(n)$. Combined with Theorem~\ref{thm:SleatorT85} and the fact that each cycle has at most $n$ nodes, we get that the total time is $\Oh(n + m \log n)$, where $\Oh(n)$ owes to the initial construction of the splay trees. 
\end{proof}

\begin{figure}[t]
    \centering
    \tikzset{every picture/.style={line width=0.75pt}} 
    \begin{tikzpicture}[x=0.75pt,y=0.75pt,yscale=-1,xscale=1]
    
    \draw   (79.67,25.33) .. controls (79.67,17.05) and (86.38,10.33) .. (94.67,10.33) .. controls (102.95,10.33) and (109.67,17.05) .. (109.67,25.33) .. controls (109.67,33.62) and (102.95,40.33) .. (94.67,40.33) .. controls (86.38,40.33) and (79.67,33.62) .. (79.67,25.33) -- cycle ;
    \draw   (39.67,85.5) .. controls (39.67,77.31) and (46.31,70.67) .. (54.5,70.67) .. controls (62.69,70.67) and (69.33,77.31) .. (69.33,85.5) .. controls (69.33,93.69) and (62.69,100.33) .. (54.5,100.33) .. controls (46.31,100.33) and (39.67,93.69) .. (39.67,85.5) -- cycle ;
    \draw    (81.67,130) -- (61.67,100) ;
    \draw   (31.33,130.33) -- (48,180) -- (14.67,180) -- cycle ;
    \draw    (59.67,70.33) -- (89.67,40.33) ;
    \draw    (49.67,100.33) -- (31.67,130) ;
    \draw    (99.67,40.33) -- (129.67,70.33) ;
    \draw   (83,130.33) -- (99.67,180) -- (66.33,180) -- cycle ;
    \draw   (130,70.67) -- (146.67,120.33) -- (113.33,120.33) -- cycle ;
    \draw   (169.67,90) -- (211.67,90) -- (211.67,80) -- (239.67,100) -- (211.67,120) -- (211.67,110) -- (169.67,110) -- cycle ;
    \draw   (273.33,70.33) -- (290,120) -- (256.67,120) -- cycle ;
    \draw   (322.67,130.33) -- (339.33,180) -- (306,180) -- cycle ;
    \draw   (329.33,85.17) .. controls (329.33,76.97) and (335.97,70.33) .. (344.17,70.33) .. controls (352.36,70.33) and (359,76.97) .. (359,85.17) .. controls (359,93.36) and (352.36,100) .. (344.17,100) .. controls (335.97,100) and (329.33,93.36) .. (329.33,85.17) -- cycle ;
    \draw    (339.33,100) -- (323,130) ;
    \draw    (349.33,100) -- (363,130) ;
    \draw   (363.33,130.33) -- (380,180) -- (346.67,180) -- cycle ;
    \draw   (293,25.17) .. controls (293,16.97) and (299.64,10.33) .. (307.83,10.33) .. controls (316.03,10.33) and (322.67,16.97) .. (322.67,25.17) .. controls (322.67,33.36) and (316.03,40) .. (307.83,40) .. controls (299.64,40) and (293,33.36) .. (293,25.17) -- cycle ;
    \draw    (273,70) -- (303,40) ;
    \draw    (313,40) -- (343,70) ;
    
    \draw (90,20) node [anchor=north west][inner sep=0.75pt]   [align=left] {$\displaystyle y$};
    \draw (50,80) node [anchor=north west][inner sep=0.75pt]   [align=left] {$\displaystyle x$};
    \draw (25,152.33) node [anchor=north west][inner sep=0.75pt]   [align=left] {$\displaystyle A$};
    \draw (77.33,152.33) node [anchor=north west][inner sep=0.75pt]   [align=left] {$\displaystyle B$};
    \draw (123.67,92.67) node [anchor=north west][inner sep=0.75pt]   [align=left] {$\displaystyle C$};
    \draw (267.67,92.33) node [anchor=north west][inner sep=0.75pt]   [align=left] {$\displaystyle A$};
    \draw (316.33,152.33) node [anchor=north west][inner sep=0.75pt]   [align=left] {$\displaystyle B$};
    \draw (339,80) node [anchor=north west][inner sep=0.75pt]   [align=left] {$\displaystyle y$};
    \draw (357.67,152.33) node [anchor=north west][inner sep=0.75pt]   [align=left] {$\displaystyle C$};
    \draw (303,20) node [anchor=north west][inner sep=0.75pt]   [align=left] {$\displaystyle x$};
    \draw (21,72) node [anchor=north west][inner sep=0.75pt]   [align=left] {$\displaystyle d$};
    \draw (118,12) node [anchor=north west][inner sep=0.75pt]   [align=left] {$\displaystyle e$};
    \draw (11,132) node [anchor=north west][inner sep=0.75pt]   [align=left] {$\displaystyle a$};
    \draw (98,132) node [anchor=north west][inner sep=0.75pt]   [align=left] {$\displaystyle b$};
    \draw (141,72) node [anchor=north west][inner sep=0.75pt]   [align=left] {$\displaystyle c$};
    \draw (331,12) node [anchor=north west][inner sep=0.75pt]   [align=left] {$\displaystyle d+c+1=e$};
    \draw (251,70) node [anchor=north west][inner sep=0.75pt]   [align=left] {$\displaystyle a$};
    \draw (378,130) node [anchor=north west][inner sep=0.75pt]   [align=left] {$\displaystyle c$};
    \draw (301,132) node [anchor=north west][inner sep=0.75pt]   [align=left] {$\displaystyle b$};
    \draw (366,72) node [anchor=north west][inner sep=0.75pt]   [align=left] {$\displaystyle e-d+b$};
    
    \end{tikzpicture}
    \caption{Update of the sizes of subtrees stored in nodes during an edge rotation. Sizes are represented by letters $a$, $b$, $c$, $d$, and $e$}
    \label{fig:size-update}
\end{figure}
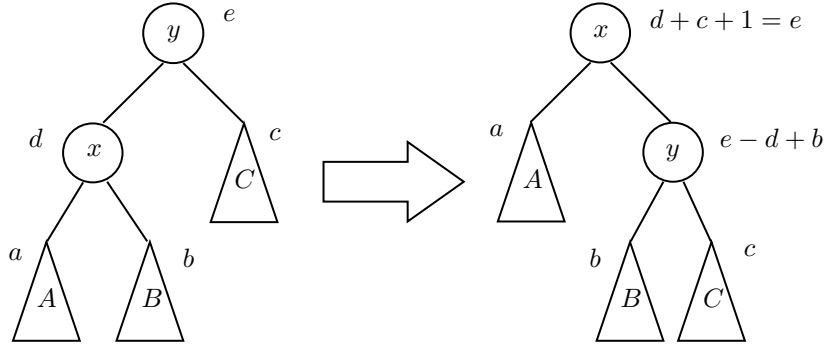


\subsection{Operations supported by the \FST}\label{sec:operations-FST}

We now describe in detail the operations that can be performed on the \FST\ data structure. 
In the following, recall that (1) in the course of each splay tree operation involving element $i$, $i$ must be splayed, (2) in the \FST, every splay tree corresponds to a cycle, keyed by position in the cycle according to some of its rotations, and (3) we have direct access to each element. Due to this direct access, we will refer to the splay tree operation {\em access}$(i,t)$ simply as {\em access}$(i)$ from now on. 

\subsubsection{Cycle rotation}

We will need a technical operation, namely for a cycle $C$ and an element $i$ in $C$, rotate $C$ such that $i$ becomes the last element in the splay tree $t$ that represents $C$. We call this operation {\em rotate$(C,i)$} and implement it as follows. 

If $i$ is the largest element of $t$, then there is nothing to do. Otherwise, let us write $C = (A, i, B)$, where $A$ is the sequence of elements that come in $C$ before $i$, and $B$ is that of those that come after $i$ in $C$. With {\em split}$(i,t)$, we turn $t$ into two trees: $t_1$, with $i$ in the root, $A$ in the left subtree and no right child; and $t_2$ representing $B$. Now we perform {\em join}$(t_2,t_1)$: this involves first splaying the largest element of $B$, say $B_{\max}$, in $t_2$, and then attaching $t_1$ (with root $i$) as the right child of $B_{\max}$. The trick is that since both $(A,i,B)$ and $(B,A,i)$ represent the same cycle $C$, the elements of $t_1$ are regarded as smaller than those of $t_2$ in the {\em split}-step but as larger in the {\em join}-step, so that we can correctly apply  {\em join$(t_2,t_1)$}. See Figure~\ref{fig:cycle-rotation} for an illustration. 

The operation consists of one splay operation---this is equivalent to {\em access}$(t,i)$ but skipping the search phase, since we have direct access to $i$---, one {\em split}, and one {\em join} operation. Therefore, the total amortized time is $\Oh(\log n)$.

\begin{figure}[t]
    \centering
    \begin{tikzpicture}[x=0.75pt,y=0.75pt,yscale=-1,xscale=1]
    
    \draw   (40,10) -- (70,103) -- (10,103) -- cycle ;
    \draw   (80,47.5) -- (116,47.5) -- (116,40) -- (140,55) -- (116,70) -- (116,62.5) -- (80,62.5) -- cycle ;
    \draw   (190,25) .. controls (190,16.72) and (196.72,10) .. (205,10) .. controls (213.28,10) and (220,16.72) .. (220,25) .. controls (220,33.28) and (213.28,40) .. (205,40) .. controls (196.72,40) and (190,33.28) .. (190,25) -- cycle ;
    \draw    (200,40) -- (180,60) ;
    \draw   (180,60) -- (200,120) -- (160,120) -- cycle ;
    \draw   (240,60) -- (260,120) -- (220,120) -- cycle ;
    \draw   (270,47.5) -- (306,47.5) -- (306,40) -- (330,55) -- (306,70) -- (306,62.5) -- (270,62.5) -- cycle ;
    \draw   (370,23.88) .. controls (370,14.97) and (377.22,7.75) .. (386.13,7.75) .. controls (395.03,7.75) and (402.25,14.97) .. (402.25,23.88) .. controls (402.25,32.78) and (395.03,40) .. (386.13,40) .. controls (377.22,40) and (370,32.78) .. (370,23.88) -- cycle ;
    \draw    (380,40) -- (360,60) ;
    \draw   (360,60) -- (380,120) -- (340,120) -- cycle ;
    \draw   (400,75) .. controls (400,66.72) and (406.72,60) .. (415,60) .. controls (423.28,60) and (430,66.72) .. (430,75) .. controls (430,83.28) and (423.28,90) .. (415,90) .. controls (406.72,90) and (400,83.28) .. (400,75) -- cycle ;
    \draw    (410,90) -- (400,110) ;
    \draw   (400,110) -- (420,170) -- (380,170) -- cycle ;
    \draw    (410,60) -- (390,40) ;
    
    \draw (33,57) node [anchor=north west][inner sep=0.75pt]   [align=left] {$C$};
    \draw (81,72) node [anchor=north west][inner sep=0.75pt]   [align=left] {split $\displaystyle (i,t)$};
    \draw (201,20) node [anchor=north west][inner sep=0.75pt]   [align=left] {$\displaystyle i$};
    \draw (173,92) node [anchor=north west][inner sep=0.75pt]   [align=left] {$A$};
    \draw (233,92) node [anchor=north west][inner sep=0.75pt]   [align=left] {$B$};
    \draw (271,72) node [anchor=north west][inner sep=0.75pt]   [align=left] {join $\displaystyle (t_2,t_1)$};
    \draw (373.5,20) node [anchor=north west][inner sep=0.75pt]   [align=left] {{\scriptsize $\text{B}_{\max}$}};
    \draw (351,92) node [anchor=north west][inner sep=0.75pt]   [align=left] {$B'$};
    \draw (411,70) node [anchor=north west][inner sep=0.75pt]   [align=left] {$\displaystyle i$};
    \draw (393,142) node [anchor=north west][inner sep=0.75pt]   [align=left] {$A$};

    \end{tikzpicture}
    \caption{Cycle-rotation operation: {\em rotate}$(C,i)$ moves $i$ to the right end of $C$. Before the rotation: $C = (A,i,B) = (A,i,B',B_{\max})$, and after the rotation $C = (B,A,i)$. }
	\label{fig:cycle-rotation}
\end{figure}
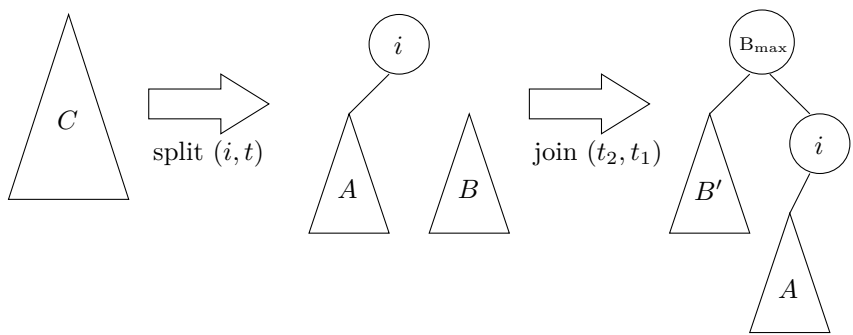
\tikzset{every picture/.style={line width=0.75pt}} 


\subsubsection{Return $\pi(i)$ and $\pi^{-1}(j)$}\label{sec:pi-tree}

Since $\pi$ is stored in form of its cycles, $\pi(i)$ is the element following $i$ in its cycle. This is the next larger element in the cycle, viewed cyclically; in other words, if $i$ is the largest node in its splay tree, then we have to return the smallest node; otherwise we have to return the successor node. 

To do this, we first splay $i$, moving it to the root of its tree. If it has no right child, then it is the largest node and we descend left from the root as long as there is a left child, thus returning the smallest node. Otherwise, from the root $i$ we move to the right child and then descend to the left as long as possible; this gives the successor of $i$ in the tree. In both cases, we splay the node corresponding to $\pi(i)$ after we find it. See Figure~\ref{fig:pi_i} for an example.

    \input{example_pi_i}

In terms of splay tree primitives, our operation is equivalent to {\em access}$(i)$ (search for $i$ in its tree, then splay $i$), followed by {\em access}$(\pi(i))$ (search for $\pi(i)$ in the splay tree, then splay $\pi(i)$), except that we are skipping the first part of the {\em access}$(i)$ operation due to our direct access to node $i$. Note that both searches are for the key (i.e., the position of the element in its cycle), and not for the element itself. 
The total amortized time is then $\Oh(\log n)$.

Finding $\pi^{-1}(j)$ is analogous, except that now we need the predecessor in the cycle rather than the successor. We splay node $j$, go left once, then keep going to the right, and finally splay the node that we find at the end of the operation. Again, we have a boundary case: if $i$ the leftmost element of the cycle, then we go straight to the rightmost element in the tree. 
The cost is $\Oh(\log n)$ amortized time, equivalent to {\em access}$(j)$ followed by {\em access}$(\pi^{-1}(j))$. 


\subsubsection{Return the number of cycles}

This operation takes constant time since we just return the contents of counter {\em cycles}. 


\subsubsection{Same-cycle query}

Given $i\neq j$, {\em samecycle}$(i,j)$ returns TRUE if and only if $i$ and $j$ are in the same cycle. We do this by first splaying $i$ (after direct access to it), thus moving it to the root of its tree, and then splaying $j$ (after direct access to it), thus moving it to the root of its tree. Now we check whether $i$ has a parent, and return TRUE if the answer is yes, since this means that $j$ has now replaced $i$ as root of their common tree. Since this is equivalent to {\em access}$(i)$ followed by {\em access}$(j)$ (in both cases skipping the search phase), the amortized running time is $\Oh(\log n)$. 


\subsubsection{Return $\pi^k(i)$ and $\pi^{-k}(i)$} \label{sec:pi-k}

Note that if $i$ is in position $j$ in its cycle $C$, then $\pi^k(i)$ is in position $(j+k) \bmod |C|$. In other words, we need to find the 
 $k'$th successor of $i$ in $C$, where $k' = k \bmod |C|$. After splaying $i$, the $k'$th successor is in the right subtree if $j+k' \leq |C|$, and in the left subtree otherwise. 

So we first splay $i$, then check if $k'\leq$ size(right($i$)). If so, then we have to return the $k'$-th smallest element in the right subtree, otherwise we need to return the $(k'-$size(right($i$)))-th smallest element in the whole tree. This can be done using a standard function on binary search trees: the function \text{min}($x,\ell$) returns the $\ell$th smallest element in the subtree rooted in $x$, defined recursively as follows: 
\begin{equation*}
    \text{min}(x,\ell) =
    \begin{cases}
        \text{return } x &\text{if } \ell = \text{size(left($x$))+1} \\
        \text{min}(\text{left}(x),\ell) &\text{if } \ell < \text{size(left($x$))+1} \\
        \text{min}(\text{right}(x),\ell-(\text{size(left($x$))+1})) &\text{if } \ell > \text{size(left($x$))+1} 
    \end{cases}
\end{equation*}

We return \text{min}(right($i$),$k'$) if $k'\leq$ size(right($i$)), and \text{min}($i$,$k'-$ size(right($i$)) otherwise. After finding $\pi^k(i)$, we splay it.
As the series of operations corresponds to {\em access}$(i)$ (without the search phase) and {\em access}($\pi^k(i)$), 
the total amortized running time is $\Oh(\log n)$. 

Returning $\pi^{-k}(i)$ is analogous, where we set $k' = -k \bmod |C|$, thus again we obtain $\Oh(\log n)$ amortized time. 


\subsubsection{Distance between two elements $i,j$}

The distance between two elements $i, j$ given by $\pi$ is $\dist_{\pi}(i,j) = \min\{d\geq 0 : \pi^d(i) = j\}$; in particular, $\dist_{\pi}(i,j) = \infty$ if no such $d$ exists.  
Clearly, $\dist_{\pi}(i,j)$ is finite if and only if $i$ and $j$ are in the same cycle. 
Note that $\dist_{\pi}$ is not symmetric. 

To compute $\dist_{\pi}(i,j)$, we first execute a query {\em samecycle}$(i,j)$ and return $\infty$ if the answer is FALSE. Otherwise, let $C$ be the cycle containing $i$ and $j$. We move $j$ to the end of $C$ with {\em rotate}$(C,j)$. Then we splay $i$ and return $\text{size(right}(i)))$. To see that this is correct, notice that splaying does not change the relative positions of the elements of the cycle. 

All the involved operations take $\Oh(\log n)$ amortized time ({\em samecycle, rotate, splay}), so this is also the cost of this operation. 

\subsubsection{Size of the cycle of an element $i$}

This can be computed by accessing $i$, splaying it so that it becomes the root of its splay tree, and then returning its subtree size. The operation takes $\Oh(\log n)$ amortized time.

\subsubsection{Update: transpositions $(\pi(i),\pi(j))$ and $(i,j)$}\label{sec:transp}

The application of the transposition $(\pi(i),\pi(j))$ results in a new permutation $\pi' = (\pi(i),\pi(j))\cdot\pi$. Lemma~\ref{lemma:cycles} 
gives the exact form of $\pi'$, depending on whether $i$ and $j$ are in the same cycle. 

So first we need to do a {\em samecycle$(i,j)$} check. If the answer is TRUE, let $C$ be the cycle containing both $i$ and $j$. We move $i$ to the end of $C$ with {\em rotate}$(C,i)$. Now we have $C = (A,j,B,i)$, and by Lemma~\ref{lemma:cycles}, $C$ will be split into $C_1 = (A,j)$ and $C_2 = (B,i)$, which can be implemented as the splay tree operation {\em split$(j,t)$}, where $t$ is the splay tree of $C$. 

Otherwise, let $C=(A,i,B)$ be the cycle containing $i$ and $C'=(D,j,E)$ the cycle containing $j$. We perform two cycle rotations, {\em rotate}$(C,i)$ and {\em rotate}$(C',j)$, moving $i$ to the end of $C$ and $j$ to the end of $C'$. Let the two trees be $t_1$ and $t_2$. The last step is to merge these two trees with {\em join$(t_1,t_2)$}, which results in $i$, the largest element of $t_1$, being splayed and $t_2$ being attached as $i$'s right child. The merged cycle represents $(B,A,i,E,D,j)$, in agreement with Lemma~\ref{lemma:cycles}. See Figure~\ref{fig:merge} for an illustration.

Note that in both cases, the {\em cycles} counter has to be updated: incremented by one if {\em samecycle$(i,j)$} is TRUE, since a {\em split}-operation is performed; and decremented by one otherwise, since a {\em join}-operation is performed. 

For the analysis, we have applied a same cycle query, followed by splitting a cycle or merging two cycles. Splitting a cycle consists of one cycle rotation and one {\em split}-operation. Merging two cycles, of two cycle rotations and one {\em join}-operation. Altogether we have, in both cases,  $\Oh(\log n)$ amortized time. 

\input{merge}



For $i, j \in \{1, 2, \ldots, n\}, i \neq j$, the transposition $(i,j)$ is equivalent to $(\pi(\pi^{-1}(i)),\pi(\pi^{-1}(j))) \circ \pi$. So we can access $\pi^{-1}(i)$ and $\pi^{-1}(j)$ in $\Oh(\log n)$ time, and then perform the transposition as in Section~\ref{sec:transp}. Overall, this takes $\Oh(\log n)$ amortized time.

\subsubsection{Update: Flips}

Our final operation, that is strongly related to the cycle structure of the permutations, is to reverse part of a cycle. Given a cycle $C = (i_1,i_2,\ldots,i_\ell)$ of $\pi$, the operation \emph{flip}$(i_r,i_t)$, with $r<t$, converts the cycle into 
$(i_1,\ldots,i_{r-1},i_t,i_{t-1},\ldots,i_r,i_{t+1},\ldots,i_\ell)$. That is, the 
direction of the cycle segment between $i_r$ and $i_t$ is reversed. It might also be that $r>t$, which yields $(i_t,i_{t-1},\ldots,i_1,i_\ell,\ldots,i_r,i_{t+1},\ldots,i_{r-1})$. 

Note that this operation is distinct from what is called a ``reversal'' in the area of genome rearrangements~\cite{Gus97}, because reversals act on the one-line representation of the permutation. 

Reflecting this operation in our current structure requires $\Oh(n)$ time, because potentially large parts of a splay tree need to be reversed. Instead, we extend our FST data structure so that the subtree-size component becomes {\em signed}. If the subtree-size field of a node $v$ is $-s$, with $s > 0$, this means that the actual subtree size is $s$ and that its subtree should be reversed, that is, its nodes should be read right-to-left. We will de-amortize the reversal work along future visits to the subtree, as explained soon. The extra space required is just $n$ bits for the signs, so it stays within $3n\log n + \Oh(n)$ bits.

In order to apply the described flip, we first perform {\em rotate}$(C,i_{t+1})$ to make sure that $i_{r-1}$ is behind it in the tree. We now splay $i_{t+1}$ and then $i_{r-1}$. After this, $i_{r-1}$ is the root of the tree, $i_{t+1}$ is its right child, and the left child $v$ of $i_{t+1}$ is the subtree with all the elements from $i_r$ to $i_t$. We then toggle the sign of the subtree-size field of $v$ and finish.

This takes $\Oh(\log n)$ amortized time because it builds on a constant number of other operations we have already analyzed. We must, however, adapt all the other operations to handle negative subtree-size fields. 

The general solution is that {\em every time we access a tree node, if its subtree-size field is negative, we toggle it, exchange the left and right children, and toggle their subtree-size fields, before proceeding with any other action.} Precisely, we define the primitive {\em fix}$(x)$ as follows: (i) if $\text{size}(x) < 0$, then (ii) toggle $\text{size}(x) \gets -\text{size}(x)$, $\text{size(left}(x)) \gets -\text{size(left}(x))$, $\text{size(right}(x)) \gets -\text{size(right}(x))$, and (iii) swap left$(x)$ with right$(x)$. See Figure~\ref{fig:fix-operations} for an example of {\em fix}$(x)$.
We then alter the splay and tree traversal operations as follows:
\begin{itemize}
\item Before performing a rotation on node $x$ during a splay, we fix the grandparent of $x$, then its parent, and then $x$. The order is important because fixing a node may change the signs of its children. The other subtrees involved in the rotations can be left unfixed. Then we perform the zig, zig-zig, or zig-zag to move $x$ upwards, as it corresponds.
\item When we descend in the tree from a node $x$ (e.g., in the function $\min(x,\ell)$), we perform {\em fix}$(x)$ before processing it.
\end{itemize}

\begin{figure}
    \centering
    
\resizebox{\textwidth}{!}{
    \tikzset{every picture/.style={line width=0.75pt}} 
    
    \begin{tikzpicture}[x=0.75pt,y=0.75pt,yscale=-1,xscale=1]
    
    \draw   (70,18.5) .. controls (70,12.15) and (75.15,7) .. (81.5,7) .. controls (87.85,7) and (93,12.15) .. (93,18.5) .. controls (93,24.85) and (87.85,30) .. (81.5,30) .. controls (75.15,30) and (70,24.85) .. (70,18.5) -- cycle ;
    \draw   (50,51) -- (70,120) -- (30,120) -- cycle ;
    \draw    (80,30) -- (50,51) ;
    \draw   (110,51) -- (130,120) -- (90,120) -- cycle ;
    \draw    (81.5,30) -- (110,51) ;
    \draw   (150,60) -- (192,60) -- (192,50) -- (220,70) -- (192,90) -- (192,80) -- (150,80) -- cycle ;
    \draw   (270,18.5) .. controls (270,12.15) and (275.15,7) .. (281.5,7) .. controls (287.85,7) and (293,12.15) .. (293,18.5) .. controls (293,24.85) and (287.85,30) .. (281.5,30) .. controls (275.15,30) and (270,24.85) .. (270,18.5) -- cycle ;
    \draw   (250,51) -- (270,120) -- (230,120) -- cycle ;
    \draw    (280,30) -- (250,51) ;
    \draw   (310,51) -- (330,120) -- (290,120) -- cycle ;
    \draw    (281.5,30) -- (310,51) ;
    \draw   (351,60) -- (393,60) -- (393,50) -- (421,70) -- (393,90) -- (393,80) -- (351,80) -- cycle ;
    \draw   (471,18.5) .. controls (471,12.15) and (476.15,7) .. (482.5,7) .. controls (488.85,7) and (494,12.15) .. (494,18.5) .. controls (494,24.85) and (488.85,30) .. (482.5,30) .. controls (476.15,30) and (471,24.85) .. (471,18.5) -- cycle ;
    \draw   (451,51) -- (471,120) -- (431,120) -- cycle ;
    \draw    (481,30) -- (451,51) ;
    \draw   (511,51) -- (531,120) -- (491,120) -- cycle ;
    \draw    (482.5,30) -- (511,51) ;
    
    \draw (76,14) node [anchor=north west][inner sep=0.75pt]   [align=left] {$\displaystyle x$};
    \draw (44,91) node [anchor=north west][inner sep=0.75pt]   [align=left] {$\displaystyle A$};
    \draw (104,91) node [anchor=north west][inner sep=0.75pt]   [align=left] {$\displaystyle B$};
    \draw (33,42) node [anchor=north west][inner sep=0.75pt]  [color={rgb, 255:red, 255; green, 0; blue, 0 }  ,opacity=1 ] [align=left] {$\displaystyle a$};
    \draw (118,42) node [anchor=north west][inner sep=0.75pt]  [color={rgb, 255:red, 255; green, 0; blue, 0 }  ,opacity=1 ] [align=left] {$\displaystyle b$};
    \draw (100,10) node [anchor=north west][inner sep=0.75pt]  [color={rgb, 255:red, 255; green, 0; blue, 0 }  ,opacity=1 ] [align=left] {$\displaystyle -c$};
    \draw (151,95) node [anchor=north west][inner sep=0.75pt]  [font=\footnotesize] [align=left] {toggle sizes};
    \draw (276,14) node [anchor=north west][inner sep=0.75pt]   [align=left] {$\displaystyle x$};
    \draw (244,91) node [anchor=north west][inner sep=0.75pt]   [align=left] {$\displaystyle A$};
    \draw (304,91) node [anchor=north west][inner sep=0.75pt]   [align=left] {$\displaystyle B$};
    \draw (229,42) node [anchor=north west][inner sep=0.75pt]  [color={rgb, 255:red, 255; green, 0; blue, 0 }  ,opacity=1 ] [align=left] {$\displaystyle -a$};
    \draw (318,42) node [anchor=north west][inner sep=0.75pt]  [color={rgb, 255:red, 255; green, 0; blue, 0 }  ,opacity=1 ] [align=left] {$\displaystyle -b$};
    \draw (300,10) node [anchor=north west][inner sep=0.75pt]  [color={rgb, 255:red, 255; green, 0; blue, 0 }  ,opacity=1 ] [align=left] {$\displaystyle c$};
    \draw (352,95) node [anchor=north west][inner sep=0.75pt]  [font=\footnotesize] [align=left] {swap left-right};
    \draw (477,14) node [anchor=north west][inner sep=0.75pt]   [align=left] {$\displaystyle x$};
    \draw (445,91) node [anchor=north west][inner sep=0.75pt]   [align=left] {$\displaystyle B$};
    \draw (505,91) node [anchor=north west][inner sep=0.75pt]   [align=left] {$\displaystyle A$};
    \draw (430,42) node [anchor=north west][inner sep=0.75pt]  [color={rgb, 255:red, 255; green, 0; blue, 0 }  ,opacity=1 ] [align=left] {$\displaystyle -b$};
    \draw (519,42) node [anchor=north west][inner sep=0.75pt]  [color={rgb, 255:red, 255; green, 0; blue, 0 }  ,opacity=1 ] [align=left] {$\displaystyle -a$};
    \draw (501,10) node [anchor=north west][inner sep=0.75pt]  [color={rgb, 255:red, 255; green, 0; blue, 0 }  ,opacity=1 ] [align=left] {$\displaystyle c$};

    \end{tikzpicture}
}
    \caption{Example of fix operation on node $x$. Note that $a$ and/or $b$ could have been negative.}
    \label{fig:fix-operations}
\end{figure}
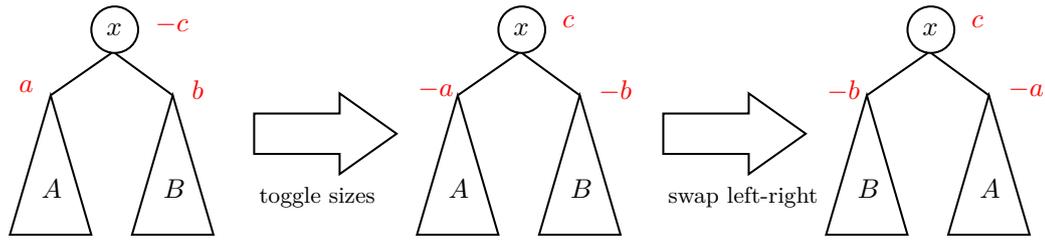

Note that {\em fix} takes constant time and does not change the potential function $\phi$, so no time complexities change due to our adjustments. All the structural changes to the splay tree are performed over traditional (i.e., fixed) nodes, so no algorithm needs further changes.


\section{Comparison with other data structures} \label{sec:comparison}

We now compare the running times of different operations of the \FST, to four baselines: (1) one array for the one-line notation, (2) two arrays for the one-lines of the permutation and its inverse, (3) the dynamic sequence representation of Munro and Nekrich \cite{MunroN15}, and (4) the static structure of Munro et al.~\cite{MRRR12}.

The \FST\ takes $3n\log n+\Oh(n)$ bits of space for the matrix $M$ and the counter {\em cycles}. The one-line notation is an integer array, taking $n\lceil\log n\rceil = n\log n + \Oh(n)$ bits. The permutation and its inverse in one-line notation require $2n\lceil\log n\rceil = 2n\log n + \Oh(n)$ bits. The dynamic sequence representation takes $n\log n + o(n\log n)$ bits. Finally, we will use the variant of the structure of Munro et al.\ that uses $(1+\epsilon)n\log n$ bits, for any constant $\epsilon > 0$.

In the array for the one-line notation (1), the update by transposition $(\pi(i),\pi(j))$ and returning $\pi(i)$ take constant time. To enable also the update by $(i,j)$ and returning $\pi^{-1}(i)$ in constant time, we need also the inverse permutation stored in another array (2); otherwise we need $\Oh(c)$ time, on a cycle of length $c$, to find the inverses of $i$ and $j$. For both structures, returning the number of cycles after an update operation is not constant, because at each update we lose the information about cycles. Even if we used a counter for storing the number of cycles, as in the forest of splay trees, we would need $\Oh(n)$ time to update it. The main problem is how to answer the same-cycle query during each update of the permutation, which takes $\Oh(c)$ time. Also, computing the distance between two elements $i,j$ takes $\Oh(c)$ time. A flip can also be applied in $\Oh(c)$ time by following the cycle.

The dynamic sequence representation (3) provides access to both $\pi(i)$ and $\pi^{-1}(j)$ within less space, in $\Oh(\frac{\log n}{\log\log n})$ time. It can also implement both transpositions by means of inserting and deleting two pairs of symbols in the sequence, within the same time complexity. The first pair would delete $i$ (or $\pi(i)$) and insert $j$ (or $\pi(j)$) in its position, while the second couple works symmetrically.
Just as the simpler preceding structures, however, we cannot answer queries related to cycles in less than $\Oh(c)$ steps, each taking $\Oh(\frac{\log n}{\log\log n})$ time.

Munro et al.~\cite{MRRR12} present a representation specialized in answering powers of permutations (4). From its description in Section~\ref{sec:related}, it follows that they support $\pi(i)$, $\pi^{-1}(j)$, $\pi^k(k)$, and $\pi^{-k}(j)$, all in time $\Oh(1/\epsilon)$. It is not hard to see that this structure can also determine if $i$ and $j$ are in the same cycle, by checking whether the 1s preceding and following $\rho^{-1}(i)$ and $\rho^{-1}(j)$ are the same; their distance in the cycle and the size of their cycle are also trivially found. The number of cycles (or 1s in the bitvector) can be stored to answer that query in constant time. Overall, they support all the queries in time $O(1/\epsilon)$. The problem is that this structure is static, so any update requires reconstructing the whole structure in $\Oh(n)$ time. 

Table~\ref{table:operations} summarizes the space and running times. The FST is the only dynamic structure that efficiently handles  queries about the cycle structure of the underlying permutation.

\begin{table}[t]
\footnotesize
\begin{tabular}{ r | c | c | c | c | c}
	& \FST\ & One-line & One-line$+$inv. & Dyn.\ Seq. & Static \\
	\hline
Space (bits) & $3n\log n$ & $n\log n$ & $2n\log n$ & $n \log n$ & $(1+\epsilon)n\log n$ \\
	\hline
	Return $\pi(i)$ & am.\ $\Oh(\log n)$ & $\Oh(1)$ & $\Oh(1)$ & $\Oh(\frac{\log n}{\log\log n})$ & $\Oh(1/\epsilon)$ \\
	Return $\pi^{-1}(j)$ & am.\ $\Oh(\log n)$ & $\Oh(c)$ & $\Oh(1)$ & $\Oh(\frac{\log n}{\log\log n})$ & $\Oh(1/\epsilon)$ \\
	Return $\pi^k(i)$ & am.\ $\Oh(\log n)$ & $\Oh(k)$ & $\Oh(k)$ & $\Oh(k \cdot \frac{\log n}{\log\log n})$ & $\Oh(1/\epsilon)$ \\
	Return $\pi^{-k}(j)$ & am.\ $\Oh(\log n)$ & $\Oh(c)$ & $\Oh(k)$  & $\Oh(k \cdot \frac{\log n}{\log\log n})$ & $\Oh(1/\epsilon)$ \\
 \hline
	  Number of cycles & $\Oh(1)$ & $\Oh(n)$ & $\Oh(n)$ & $\Oh(n)$ & $\Oh(1)$ \\
        Size of $i$'s cycle & am.\ $\Oh(\log n)$ & $\Oh(c)$ & $\Oh(c)$ & $\Oh(c\cdot \frac{\log n}{\log\log n})$ & $\Oh(1/\epsilon)$ \\	
        Same cycle $i,j$ & am.\ $\Oh(\log n)$ & $\Oh(c)$ & $\Oh(c)$ & $\Oh(c\cdot \frac{\log n}{\log\log n})$ & $\Oh(1/\epsilon)$ \\
	Cycle distance $i,j$ & am.\ $\Oh(\log n)$ & $\Oh(c)$ & $\Oh(c)$ & $\Oh(c\cdot \frac{\log n}{\log\log n})$ & $\Oh(1/\epsilon)$ \\
 \hline
	Transp.\ $(\pi(i),\pi(j))$ & am.\ $\Oh(\log n)$ & $\Oh(1)$ & $\Oh(1)$ & $\Oh(\frac{\log n}{\log\log n})$ & $\Oh(n)$ \\
	Transp.\ $(i, j)$ & am.\ $\Oh(\log n)$ & $\Oh(c)$ & $\Oh(1)$ & $\Oh(\frac{\log n}{\log\log n})$ & $\Oh(n)$ \\
        Flip $i,j$     & am.\ $\Oh(\log n)$ & $\Oh(c)$ & $\Oh(c)$ & $\Oh(c \cdot \frac{\log n}{\log\log n})$ & $\Oh(n)$ \\
    \hline	
\end{tabular}
\caption{Comparison between the forest of splay trees and several baselines to handle queries and updates on permutations. We use ``am.'' as a shorthand for ``amortized'', $c$ for the length of the cycle where $i$ belongs ($c=n$ in the worst case), and write only the leading term of the space.}
\label{table:operations}
\end{table}


\section{Conclusion} \label{sec:conclusion}

We have introduced a new dynamic data structure to represent permutations, the forest of splay trees (FST), which is unique in supporting various operations related to the cycle structure of the permutation, while permitting to perform arbitrary transpositions on it and flips in the cycles. Concretely, for a permutation on $n$, the FST is built in $\Oh(n)$ time, and then supports a number of queries and updates in $\Oh(\log n)$ amortized time each. No structure we know of supports both kinds of queries/updates in $o(n)$ time.

A future direction to extend the FST is to incorporate other queries and updates, motivated by applications. 
Another interesting direction is to extend the scope of FSTs from permutations to general functions in $[1,n]$, in the lines of the representation of Munro et al.~\cite{MRRR12}. They combine cycles with ordinal trees in order to support the operations $f^k(i)$ and $f^{-k}(\{j\})$ in optimal time, but again, do not support updates on $f$.

\bibliography{biblio}

\appendix 

\section*{APPENDIX}

\section{Details on splay trees}

These chains of rotations have three different names based on the relative position of a node $x$ w.r.t. its parent and grandparent. If the parent of $x$ is the root, then only one rotation is required to move $x$ to the root (\textit{zig}). If both $x$ and the parent of $x$ are right children of their parent, or if both are left children of their parent, then two rotations are performed in sequence: first between the parent of $x$ and the grandparent of $x$, then between $x$ and its parent (\textit{zig-zig}, see Figure~\ref{fig:zig-zig}). The last possibility is that $x$ is a right child and its parent is a left child of the grandparent of $x$, or that $x$ is a left child and its parent is a right child. Two rotations with different direction are concatenated: this time we first perform a rotation between $x$ and its parent, then between $x$ and its former grandparent (\textit{zig-zag}, see Figure~\ref{fig:zig-zag}). 

\begin{figure}[ht]

\resizebox{\textwidth}{!}{
\tikzset{every picture/.style={line width=0.75pt}} 

\begin{tikzpicture}[x=0.75pt,y=0.75pt,yscale=-1,xscale=1]

\draw   (86.33,15.17) .. controls (86.33,6.97) and (92.97,0.33) .. (101.17,0.33) .. controls (109.36,0.33) and (116,6.97) .. (116,15.17) .. controls (116,23.36) and (109.36,30) .. (101.17,30) .. controls (92.97,30) and (86.33,23.36) .. (86.33,15.17) -- cycle ;
\draw   (56,75) .. controls (56,66.72) and (62.72,60) .. (71,60) .. controls (79.28,60) and (86,66.72) .. (86,75) .. controls (86,83.28) and (79.28,90) .. (71,90) .. controls (62.72,90) and (56,83.28) .. (56,75) -- cycle ;
\draw   (26.33,135.17) .. controls (26.33,126.97) and (32.97,120.33) .. (41.17,120.33) .. controls (49.36,120.33) and (56,126.97) .. (56,135.17) .. controls (56,143.36) and (49.36,150) .. (41.17,150) .. controls (32.97,150) and (26.33,143.36) .. (26.33,135.17) -- cycle ;
\draw    (96,30) -- (76,60) ;
\draw    (96,120) -- (76,90) ;
\draw   (16.67,180.33) -- (33.33,230) -- (0,230) -- cycle ;
\draw    (46,120) -- (66,90) ;
\draw    (106,30) -- (126,60) ;
\draw   (126.33,60.33) -- (143,110) -- (109.67,110) -- cycle ;
\draw    (36,150) -- (16,180) ;
\draw    (46,150) -- (66,180) ;
\draw   (96.33,120.33) -- (113,170) -- (79.67,170) -- cycle ;
\draw   (66.33,180.33) -- (83,230) -- (49.67,230) -- cycle ;
\draw   (163,110) -- (205,110) -- (205,100) -- (233,120) -- (205,140) -- (205,130) -- (163,130) -- cycle ;
\draw   (268.33,124.83) .. controls (268.33,116.64) and (274.97,110) .. (283.17,110) .. controls (291.36,110) and (298,116.64) .. (298,124.83) .. controls (298,133.03) and (291.36,139.67) .. (283.17,139.67) .. controls (274.97,139.67) and (268.33,133.03) .. (268.33,124.83) -- cycle ;
\draw    (278,140) -- (258,170) ;
\draw    (288,140) -- (308,170) ;
\draw   (257.67,169.33) -- (274.33,219) -- (241,219) -- cycle ;
\draw   (308.33,170) -- (325,219.67) -- (291.67,219.67) -- cycle ;
\draw   (368,124.83) .. controls (368,116.64) and (374.64,110) .. (382.83,110) .. controls (391.03,110) and (397.67,116.64) .. (397.67,124.83) .. controls (397.67,133.03) and (391.03,139.67) .. (382.83,139.67) .. controls (374.64,139.67) and (368,133.03) .. (368,124.83) -- cycle ;
\draw    (378,140) -- (358,170) ;
\draw    (388,140) -- (408,170) ;
\draw   (357.67,170) -- (374.33,219.67) -- (341,219.67) -- cycle ;
\draw   (408.33,170) -- (425,219.67) -- (391.67,219.67) -- cycle ;
\draw   (320,64.83) .. controls (320,56.64) and (326.64,50) .. (334.83,50) .. controls (343.03,50) and (349.67,56.64) .. (349.67,64.83) .. controls (349.67,73.03) and (343.03,79.67) .. (334.83,79.67) .. controls (326.64,79.67) and (320,73.03) .. (320,64.83) -- cycle ;
\draw    (330,80) -- (290,110) ;
\draw    (340,80) -- (380,110) ;
\draw   (429,110.33) -- (471,110.33) -- (471,100.33) -- (499,120.33) -- (471,140.33) -- (471,130.33) -- (429,130.33) -- cycle ;
\draw   (541,15.17) .. controls (541,6.97) and (547.64,0.33) .. (555.83,0.33) .. controls (564.03,0.33) and (570.67,6.97) .. (570.67,15.17) .. controls (570.67,23.36) and (564.03,30) .. (555.83,30) .. controls (547.64,30) and (541,23.36) .. (541,15.17) -- cycle ;
\draw    (551,30) -- (531,60) ;
\draw    (561,30) -- (581,60) ;
\draw   (530.67,60.33) -- (547.33,110) -- (514,110) -- cycle ;
\draw   (560.67,120.33) -- (577.33,170) -- (544,170) -- cycle ;
\draw   (603,135.17) .. controls (603,126.97) and (609.64,120.33) .. (617.83,120.33) .. controls (626.03,120.33) and (632.67,126.97) .. (632.67,135.17) .. controls (632.67,143.36) and (626.03,150) .. (617.83,150) .. controls (609.64,150) and (603,143.36) .. (603,135.17) -- cycle ;
\draw    (613,150.33) -- (593,180.33) ;
\draw    (623,150.33) -- (643,180.33) ;
\draw   (592.67,180.33) -- (609.33,230) -- (576,230) -- cycle ;
\draw   (643.33,180.33) -- (660,230) -- (626.67,230) -- cycle ;
\draw   (571,75.17) .. controls (571,66.97) and (577.64,60.33) .. (585.83,60.33) .. controls (594.03,60.33) and (600.67,66.97) .. (600.67,75.17) .. controls (600.67,83.36) and (594.03,90) .. (585.83,90) .. controls (577.64,90) and (571,83.36) .. (571,75.17) -- cycle ;
\draw    (591,90.33) -- (611,120) ;
\draw    (581,90) -- (561,120) ;

\draw (97.33,9) node [anchor=north west][inner sep=0.75pt]   [align=left] {$g$};
\draw (67.67,69) node [anchor=north west][inner sep=0.75pt]   [align=left] {$p$};
\draw (37.33,129) node [anchor=north west][inner sep=0.75pt]   [align=left] {$x$};
\draw (10.33,202.33) node [anchor=north west][inner sep=0.75pt]   [align=left] {$A$};
\draw (120,82.33) node [anchor=north west][inner sep=0.75pt]   [align=left] {$D$};
\draw (90.67,142.33) node [anchor=north west][inner sep=0.75pt]   [align=left] {$C$};
\draw (60,202.33) node [anchor=north west][inner sep=0.75pt]   [align=left] {$B$};
\draw (161,142) node [anchor=north west][inner sep=0.75pt]   [align=left] {rotate p-g};
\draw (279.33,119) node [anchor=north west][inner sep=0.75pt]   [align=left] {$x$};
\draw (250,191.33) node [anchor=north west][inner sep=0.75pt]   [align=left] {$A$};
\draw (302,192) node [anchor=north west][inner sep=0.75pt]   [align=left] {$B$};
\draw (379,119) node [anchor=north west][inner sep=0.75pt]   [align=left] {$g$};
\draw (352,192) node [anchor=north west][inner sep=0.75pt]   [align=left] {$C$};
\draw (402,192) node [anchor=north west][inner sep=0.75pt]   [align=left] {$D$};
\draw (331,59) node [anchor=north west][inner sep=0.75pt]   [align=left] {$p$};
\draw (427,142.33) node [anchor=north west][inner sep=0.75pt]   [align=left] {rotate p-g};
\draw (552,9) node [anchor=north west][inner sep=0.75pt]   [align=left] {$x$};
\draw (525,82.33) node [anchor=north west][inner sep=0.75pt]   [align=left] {$A$};
\draw (554.33,142.33) node [anchor=north west][inner sep=0.75pt]   [align=left] {$B$};
\draw (614,129) node [anchor=north west][inner sep=0.75pt]   [align=left] {$g$};
\draw (587,202.33) node [anchor=north west][inner sep=0.75pt]   [align=left] {$C$};
\draw (637,202.33) node [anchor=north west][inner sep=0.75pt]   [align=left] {$D$};
\draw (582,69) node [anchor=north west][inner sep=0.75pt]   [align=left] {$p$};

\end{tikzpicture}
}

\caption{Example of zig-zig subroutine applied on node $x$.}
\label{fig:zig-zig}
\end{figure}
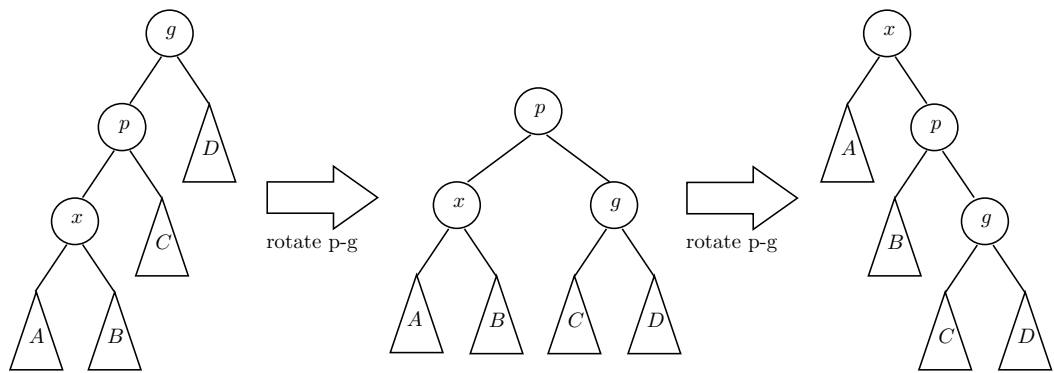

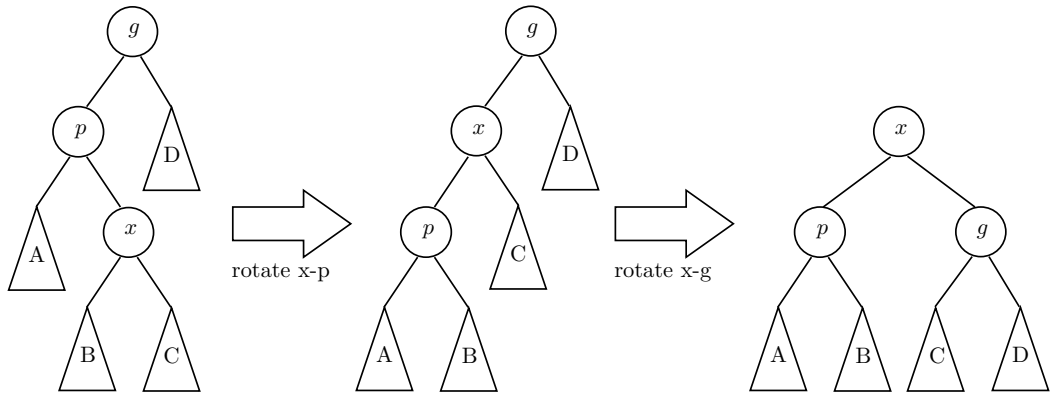
\begin{figure}

\resizebox{\textwidth}{!}{
\tikzset{every picture/.style={line width=0.75pt}} 

\begin{tikzpicture}[x=0.75pt,y=0.75pt,yscale=-1,xscale=1]

\draw   (58.67,15.17) .. controls (58.67,6.97) and (65.31,0.33) .. (73.5,0.33) .. controls (81.69,0.33) and (88.33,6.97) .. (88.33,15.17) .. controls (88.33,23.36) and (81.69,30) .. (73.5,30) .. controls (65.31,30) and (58.67,23.36) .. (58.67,15.17) -- cycle ;
\draw   (26.33,75) .. controls (26.33,66.72) and (33.05,60) .. (41.33,60) .. controls (49.62,60) and (56.33,66.72) .. (56.33,75) .. controls (56.33,83.28) and (49.62,90) .. (41.33,90) .. controls (33.05,90) and (26.33,83.28) .. (26.33,75) -- cycle ;
\draw   (56.33,135.17) .. controls (56.33,126.97) and (62.97,120.33) .. (71.17,120.33) .. controls (79.36,120.33) and (86,126.97) .. (86,135.17) .. controls (86,143.36) and (79.36,150) .. (71.17,150) .. controls (62.97,150) and (56.33,143.36) .. (56.33,135.17) -- cycle ;
\draw    (68.33,30) -- (46.33,60) ;
\draw    (66.33,120) -- (46.33,90) ;
\draw   (16.67,120) -- (33.33,169.67) -- (0,169.67) -- cycle ;
\draw    (16.33,120) -- (36.33,90) ;
\draw    (78.33,30) -- (96.33,60) ;
\draw   (96.67,60.33) -- (113.33,110) -- (80,110) -- cycle ;
\draw    (66.33,150) -- (46.33,180) ;
\draw    (76.33,150) -- (96.33,180) ;
\draw   (46.67,180) -- (63.33,229.67) -- (30,229.67) -- cycle ;
\draw   (96.67,180.33) -- (113.33,230) -- (80,230) -- cycle ;
\draw   (133,120) -- (175,120) -- (175,110) -- (203,130) -- (175,150) -- (175,140) -- (133,140) -- cycle ;
\draw   (295,14.83) .. controls (295,6.64) and (301.64,0) .. (309.83,0) .. controls (318.03,0) and (324.67,6.64) .. (324.67,14.83) .. controls (324.67,23.03) and (318.03,29.67) .. (309.83,29.67) .. controls (301.64,29.67) and (295,23.03) .. (295,14.83) -- cycle ;
\draw   (232.67,135) .. controls (232.67,126.72) and (239.38,120) .. (247.67,120) .. controls (255.95,120) and (262.67,126.72) .. (262.67,135) .. controls (262.67,143.28) and (255.95,150) .. (247.67,150) .. controls (239.38,150) and (232.67,143.28) .. (232.67,135) -- cycle ;
\draw   (262.67,74.5) .. controls (262.67,66.31) and (269.31,59.67) .. (277.5,59.67) .. controls (285.69,59.67) and (292.33,66.31) .. (292.33,74.5) .. controls (292.33,82.69) and (285.69,89.33) .. (277.5,89.33) .. controls (269.31,89.33) and (262.67,82.69) .. (262.67,74.5) -- cycle ;
\draw    (304.67,29.67) -- (282.67,59.67) ;
\draw    (272.67,180) -- (252.67,150) ;
\draw   (223,180) -- (239.67,229.67) -- (206.33,229.67) -- cycle ;
\draw    (222.67,180) -- (242.67,150) ;
\draw    (314.67,29.67) -- (332.67,59.67) ;
\draw   (333.33,60) -- (350,109.67) -- (316.67,109.67) -- cycle ;
\draw    (272.67,90) -- (252.67,120) ;
\draw    (282.67,89.33) -- (302.67,119.33) ;
\draw   (273,180.33) -- (289.67,230) -- (256.33,230) -- cycle ;
\draw   (562,134.83) .. controls (562,126.64) and (568.64,120) .. (576.83,120) .. controls (585.03,120) and (591.67,126.64) .. (591.67,134.83) .. controls (591.67,143.03) and (585.03,149.67) .. (576.83,149.67) .. controls (568.64,149.67) and (562,143.03) .. (562,134.83) -- cycle ;
\draw   (466.33,135) .. controls (466.33,126.72) and (473.05,120) .. (481.33,120) .. controls (489.62,120) and (496.33,126.72) .. (496.33,135) .. controls (496.33,143.28) and (489.62,150) .. (481.33,150) .. controls (473.05,150) and (466.33,143.28) .. (466.33,135) -- cycle ;
\draw   (513.33,74.83) .. controls (513.33,66.64) and (519.97,60) .. (528.17,60) .. controls (536.36,60) and (543,66.64) .. (543,74.83) .. controls (543,83.03) and (536.36,89.67) .. (528.17,89.67) .. controls (519.97,89.67) and (513.33,83.03) .. (513.33,74.83) -- cycle ;
\draw    (571.67,149.67) -- (549.67,179.67) ;
\draw    (506.33,180) -- (486.33,150) ;
\draw   (456.67,180) -- (473.33,229.67) -- (440,229.67) -- cycle ;
\draw    (456.33,180) -- (476.33,150) ;
\draw    (581.67,149.67) -- (599.67,179.67) ;
\draw   (600.33,180) -- (617,229.67) -- (583.67,229.67) -- cycle ;
\draw    (523.33,90) -- (483.33,120) ;
\draw    (533.33,89.67) -- (573.33,120) ;
\draw   (506.67,180.33) -- (523.33,230) -- (490,230) -- cycle ;
\draw   (360,120) -- (402,120) -- (402,110) -- (430,130) -- (402,150) -- (402,140) -- (360,140) -- cycle ;
\draw   (550,180.33) -- (566.67,230) -- (533.33,230) -- cycle ;
\draw   (302.33,119.33) -- (319,169) -- (285.67,169) -- cycle ;

\draw (69.67,9) node [anchor=north west][inner sep=0.75pt]   [align=left] {$g$};
\draw (38,69) node [anchor=north west][inner sep=0.75pt]   [align=left] {$p$};
\draw (67.33,129) node [anchor=north west][inner sep=0.75pt]   [align=left] {$x$};
\draw (10.33,142) node [anchor=north west][inner sep=0.75pt]   [align=left] {A};
\draw (90.33,82.33) node [anchor=north west][inner sep=0.75pt]   [align=left] {D};
\draw (41,202) node [anchor=north west][inner sep=0.75pt]   [align=left] {B};
\draw (90.33,202.33) node [anchor=north west][inner sep=0.75pt]   [align=left] {C};
\draw (131,152) node [anchor=north west][inner sep=0.75pt]   [align=left] {rotate x-p};
\draw (306,9) node [anchor=north west][inner sep=0.75pt]   [align=left] {$g$};
\draw (244.33,129) node [anchor=north west][inner sep=0.75pt]   [align=left] {$p$};
\draw (273.67,69) node [anchor=north west][inner sep=0.75pt]   [align=left] {$x$};
\draw (216.67,202) node [anchor=north west][inner sep=0.75pt]   [align=left] {A};
\draw (327,82) node [anchor=north west][inner sep=0.75pt]   [align=left] {D};
\draw (267.33,202.33) node [anchor=north west][inner sep=0.75pt]   [align=left] {B};
\draw (573,129) node [anchor=north west][inner sep=0.75pt]   [align=left] {$g$};
\draw (478,129) node [anchor=north west][inner sep=0.75pt]   [align=left] {$p$};
\draw (524.33,69) node [anchor=north west][inner sep=0.75pt]   [align=left] {$x$};
\draw (450.33,202) node [anchor=north west][inner sep=0.75pt]   [align=left] {A};
\draw (594,202) node [anchor=north west][inner sep=0.75pt]   [align=left] {D};
\draw (501,202.33) node [anchor=north west][inner sep=0.75pt]   [align=left] {B};
\draw (358,152) node [anchor=north west][inner sep=0.75pt]   [align=left] {rotate x-g};
\draw (544.67,202.67) node [anchor=north west][inner sep=0.75pt]   [align=left] {C};
\draw (296,141.33) node [anchor=north west][inner sep=0.75pt]   [align=left] {C};

\end{tikzpicture}
}

\caption{Example of zig-zag subroutine applied on node $x$.}
\label{fig:zig-zag}

\end{figure}

We now explain two other operations on splay trees, namely {\em split} and {\em join}. 

\paragraph*{Split operation on splay trees} Given a splay tree $t$ and an element $x$ we can split $t$ in two different trees, i.e. $t_1$ and $t_2$ s.t. every element in $t_1$ is less than or equal to $x$ and every element in $t_2$ is greater than $x$. This can be done by splaying $x$ and removing the right edge from $x$ to its right child. This operation takes amortized $\Oh(\log n)$ time, because the splay operation takes amortized $\Oh(\log n)$ time, while removing the edge takes $\Oh(1)$ time.

\paragraph*{Join operation on splay trees} Given two splay trees $t_1$ and $t_2$, we can combine them into a single tree. We have to splay the rightmost element of $t_1$ ($\max(t_1)$) (if we assume that every element of $t_1$ is smaller than the minimum of $t_2$) and then create an edge between $\max(t_1)$ and the root of $t_2$. This operation also takes amortized $\Oh(\log n)$ time (one splay operation and one edge creation). 

\section{Comparison to FST in Giuliani et al.} 

In~\cite{GiulianiLMR21}, the splay-tree based data structure was used exclusively for transpositions, and other possible operations were not discussed. Moreover, only one specific type of transposition was used: that of two contiguous elements: $\pi' = (\pi(i), \pi(i+1)) \cdot \pi$. More precisely, the authors always applied the operation $\pi' = (1, \pi(i+1)) \cdot \pi$, because in the particular application, in the given permutation $\pi_i$, 1 was always in position $i$. Another peculiarity of the application is that right after the $i$th transposition is performed, value $i+1$ is always the rightmost element of its splay tree. 

The specialized version of the \FST\ in~\cite{GiulianiLMR21} furthermore did not have the information about subtree sizes that is key to multiple operations supported in our current version. 

\end{document}